\documentclass[preprint]{elsarticle}

\usepackage{array,xspace,multirow,hhline,tikz,colortbl,tabularx,booktabs,fixltx2e,amsmath,amssymb,amsfonts,amsthm}
\usepackage{algorithm}
\usepackage{algorithmic}
\usepackage{verbatim,ifthen}
\usepackage{enumitem}
\usepackage{pifont}
\usepackage{ifthen}
\usepackage{calrsfs,mathrsfs}
\usepackage{bbding,pifont}
\usepackage{pgflibraryshapes}

\definecolor{light-gray}{gray}{0.9}

\newtheorem{definition}{Definition}%

	\newcommand{\lemref}[1]{Lemma~\ref{#1}}

	\newcommand{\ie}{i.e.,\xspace}
	\newcommand{\eg}{e.g.,\xspace}

	\newcommand{\Eg}{E.g.,\xspace}

	\newtheorem{lemma}{Lemma}%
	\newtheorem{theorem}{Theorem}%
	\newtheorem{corollary}{Corollary}%
	\newtheorem{example}{Example}

	\newcommand\eat[1]{}

	\usepackage{enumitem}
	\setenumerate[1]{label=\rm(\it{\roman{*}}\rm),ref=({\it\roman{*}}),leftmargin=*}
	\newlength{\wordlength}

	\newcommand{\set}[1]{\{#1\}}
	\newcommand{\midd}{\mathbin{:}}

	\newcommand{\eqclass}[2][]{\ifthenelse{\equal{#1}{}}{[#2]}{[#2]_{\sim_{#1}}}}

	 \newcommand{\listsize}{\ell\xspace}
	 	\newcommand{\setsize}{u\xspace}
	\newcommand{\s}{s\xspace}
	
	\newcommand{\Pref}[1][]{
		\ifthenelse{\equal{#1}{}}{\mathrel R}{\mathop{R_{#1}}}
	}                                          
	\newcommand{\sPref}[1][]{                  
		\ifthenelse{\equal{#1}{}}{\mathrel P}{\mathop{P_{#1}}}
	}                                          
	\newcommand{\Indiff}[1][]{                 
		\ifthenelse{\equal{#1}{}}{\mathrel I}{\mathop{I_{#1}}}
	}
	\newcommand{\prefset}[1][]{\ifthenelse{\equal{#1}{}}{\mathcal{R}}{\mathcal{R}_{#1}}}

\usepackage{enumitem}
\setenumerate[1]{label=\rm(\it{\roman{*}}\rm),ref=({\it\roman{*}}),leftmargin=*}

\newcommand{\nbh}[1][]{
	\ifthenelse{\equal{#1}{}}{\nu}{\nu(#1)}
}

\newcommand{\cstr}[1][]{
	\ifthenelse{\equal{#1}{}}{\mathscr S}{\cstr(#1)}
}

\newcommand{\choice}[1][]{
	\ifthenelse{\equal{#1}{}}{\mathit{C}}{\choice(#1)}
}

\begin{document}

	\title{Testing Substitutability of Weak Preferences}

	\author{Haris Aziz\corref{cor1}} \ead{aziz@in.tum.de}
	\author{Markus Brill} \ead{brill@in.tum.de}
	\author{Paul Harrenstein} \ead{harrenst@in.tum.de}

	\address{Institut f\"ur Informatik, Technische Universit\"at M\"unchen, 85748 Garching bei M\"unchen, Germany} 
	
	\cortext[cor1]{Corresponding author} 


	\begin{abstract}
	In many-to-many matching models, substitutable preferences constitute the largest domain for which a pairwise stable matching is guaranteed to exist. 
In this note, we extend the recently proposed algorithm of \citet{HIK11a} to test substitutability of weak preferences. 
	Interestingly, the algorithm is faster than the algorithm of \citeauthor{HIK11a} by a linear factor on the domain of strict preferences.
\end{abstract}


\maketitle

\section{Introduction}

In stable matching problems, the aim is to match agents in a stable manner to objects or to other agents, keeping in view the preferences of the agents involved. These problems have significant applications in matching residents to hospitals, students to schools, etc. and have received tremendous interest in mathematical economics, computer science, and operations research \citep[see \eg ][]{GuIr89a,RoSo90a}.

In many matching models individual preferences are supposed to be \emph{responsive}.  \Eg in the case in which a hospital can hire multiple doctors, the hospitals are usually and unnaturally assumed to submit 
preferences that render the choice between a pair of doctors independent of other available outcomes~\citep{HIK11a}. This assumption is rather unnatural, in particular if multiple agents can be matched to a single agent. An alternative is to allow hospitals to submit \emph{substitutable} preferences, which allows for considerably more flexibility in expressing preferences over groups of doctors.


Substitutable preferences were introduced by \citet{Roth84a} and constitute the largest domain in which stable matchings are guaranteed to exist. In many matching models, substitutability is in fact a necessary and sufficient condition for the existence of stable allocations (see footnote 4 in \citep{HIK11a}).\footnote{The settings include many-to-one matchings, many-to-many matchings, and  many-to-many matching with contracts.} The significance of substitutability leads to the natural algorithmic problem of testing  whether a given preference relation is substitutable or not. Recently, \citet{HIK11a} presented a polynomial-time algorithm for this problem.
%
Both the original definition of substitutability and the testing algorithm assume that agents express \emph{strict} preferences, \ie preferences without any indifferences.
Strict preferences are assumed for most of the results in the literature concerning substitutability. 



\citet{Soto99a} formulated a natural extension of substitutability for the more general preference domain which allows indifferences (so-called \emph{weak} preferences). 
In many settings, allowing indifferences is not only a natural relaxation but also also a practical requirement: as agents may not be able to strictly rank the respective outcomes and might be indifferent among some of them. 
The introduction of indifferences can significantly change the properties and structure of stable matchings. For example, stable matchings can have different cardinalities~\citep{MII+02a} and man or woman-optimal stable matchings are no longer well-defined for marriage markets~\citep{RoSo90a}. Weak preferences can also induce non-trivial complexity. For instance, checking whether a stable roommate matching exists is polynomial-time solvable for strict preferences~ \citep{Irvi85a} but becomes NP-complete when indifferences are allowed~\citep{Ronn90a}.


In this brief note, we examine the notion of substitutability for the general case of weak preferences.\footnote{Despite the fact that indifferences can significantly affect results in matching theory, \citet{Soto99a} showed that this generalized notion still guarantees the existence of a stable matching in many-to-many matching models.}
We identify conditions that are violated by non-substitutable preferences. Using these conditions, we obtain a polynomial-time algorithm to test substitutability of weak preferences. Restricted to the domain of strict preferences, our algorithm is faster than the algorithm of \citet{HIK11a} by a linear factor.


\section{Preliminaries}

Let $U$ be a finite set of alternatives. A (weak) preference relation $R$ is a transitive and complete relation on $2^U$. Let~$P$ and~$I$ denote the strict and symmetric parts of $R$, respectively. Each preference relation $R$ induces a choice function $C$ that returns, for each $X \subseteq U$, the set of all $R$-maximal subsets of $X$, \ie
\[
	\choice[X]=\set{Y\subseteq X\midd \text{$Y \mathrel{R} Z$ for all $Z\subseteq X$}}.
\]
A set $X\subseteq U$ is called \emph{acceptable} if $X \mathrel{R} \emptyset$. Observe that $C$ always returns at least one set (maybe the empty set) and that all sets returned by~$C$ are acceptable. 

%

The most general and expressive way of representing $R$ is via a preference list $L$ that contains all acceptable sets. This list representation is reminiscent to the \emph{representation by individually rational lists of coalitions} used in the context of hedonic coalition formation games~\citep{Ball04a}.

Let $\s$ denote the maximal size of an indifference class, where an indifference class is a family $\set{Y\in U\midd X\mathrel{I} Y}$ for some acceptable subset~$X$ in~$U$. Observe that the size of $C(\cdot)$ is bounded by $s$ and that a preference relation is \emph{strict} if and only if $s=1$. Furthermore, let~$\setsize$ denote the size of $U$ and $\listsize=|L|$ the number of acceptable sets.



\begin{example}\label{example:basic}
	%
	%
	
	Let $U=\{a,b,c,d\}$ and define the preference relation $R$ by the list
		$$\{a,b,d\}~~I~~\{b,c,d\}~~P~~\{a,b\}~~I~~\{b,c\}~~I~~\{a,c\}~~P~~\emptyset \text.$$
		
		Then, $C(U)=\{\{a,b,d\},\{b,c,d\}\}$ and $C(\{a\})=\{\emptyset\}$.
\end{example}

For a preference relation $R$ represented in list form and $X \subseteq U$, it can be checked in time $O(\listsize|X|)$ whether a given alternative is in $C(X)$.

\section{Substitutability and weak preferences}

The following definition was introduced by \citet{Soto99a}.

\begin{definition}
A preference relation~$R$ is \emph{substitutable} if and only if the following two conditions hold:
\begin{enumerate}
	\item[(S1)] for all non-empty $A,B\subseteq U$ with $B\subseteq A$ we have that for all $X\in\choice[A]$ there is some $Y\in\choice[B]$ such that $X\cap B\subseteq Y$, and
	\item[(S2)] for all non-empty $A,B\subseteq U$ with $B\subseteq A$ we have that for all $Y\in \choice[B]$ there is some $X\in\choice[A]$ such that $X\cap B\subseteq Y$.
\end{enumerate}	
\end{definition}

\begin{example}
	%
	
	Consider the preference relation $R$ from Example~\ref{example:basic}. It can be verified that $R$ satisfies (S1) and violates (S2). For the latter, take $A=U$ and $B=\{a,b,c\}$. Then, 
	$$C(B)=\{\{a,b\}, \{b,c\}, \{a,c\}\}.$$
Now $Y=\{a,c\}$ is in $\choice[B]$, but there exists no $X\in \choice[A]$ such that $X\cap B\subseteq Y$. Hence, $R$ is not substitutable.
\end{example}

The following lemma is adapted from Lemma~1 in \citep{HIK11a}.

\begin{lemma}\label{lemma:setalphahat}
	For all $A,B\subseteq U$ with $B\subseteq A$,
	\[
		\text{$\choice[A]\cap 2^B\neq\emptyset$ implies $\choice[B]=\choice[A]\cap 2^B$.}
	\]
\end{lemma}
\begin{proof}
	Assume $\choice[A]\cap 2^B\neq\emptyset$. Then, $X\in\choice[A]\cap 2^B$ for some $X\subseteq U$. First consider an arbitrary $Y\in\choice[B]$. Then $Y\mathrel R X$. Hence, $Y\in\choice[A]\cap 2^B$ as well. Now consider an arbitrary $Y\notin \choice[B]$. If $Y\notin 2^B$, immediately $Y\notin \choice[A]\cap 2^B$. If $Y\in 2^B$, we have $X\mathrel P Y$ and therefore $Y\notin\choice[A]$. Also then $Y\notin \choice[A]\cap 2^B$.
\end{proof}

\section{Testing substitutability}

We now outline a way to test substitutability of weak preferences. The idea utilizes an insight from \citep{HIK11a} that instead of checking all violations of substitutability, one may restrict one's attention to violations of a specific type. 

By an \emph{(S1)-violation for~$R$} we understand a pair $(A,B)\in 2^U\times 2^U$ such that $B\subseteq A$ and for some $X\in\choice[A]$, it is the case that $X\cap B\nsubseteq Z$ for all $Z\in\choice[B]$.

\begin{lemma}\label{lemma:S1check}
		Let $R$ be a preference relation. If there exists an (S1)-violation for~$R$, then there exist $X,Y\in L$ and $x\in X$ such that $(X\cup Y,Y\cup\set x)$ is also an (S1)-violation for~$R$.
\end{lemma}

\begin{proof}
	Assume that $(A,B)$ is an (S1)-violation for $R$. Then there is some  $X\in\choice[A]$ such that $ X\cap B\nsubseteq Z$ for all $Z\in\choice[B]$. As $\choice[B]\neq\emptyset$, there is some $Y\in\choice[B]$
such that $Y\cap X$ is maximal with respect to set inclusion, \ie $Y\cap X\subsetneq Z\cap X$ for no $Z\in\choice[B]$. Obviously, $X,Y \in L$. By our assumption, $X\cap B\nsubseteq Y$ and we may therefore assume the existence of some~$x\in X\setminus Y$.  We prove that $(X\cup Y,Y\cup\set x)$ is an (S1)-violation for~$R$, \ie
\begin{enumerate}
	\item\label{item:S1checki}   $Y\cup\set x\subseteq X\cup Y$,
	\item\label{item:S1checkii}  $X\in\choice[X\cup Y]$, and
	\item\label{item:S1checkiii} $X\cap(Y\cup\set x)\nsubseteq Z$ for all $Z\in\choice[Y\cup\set x]$.
\end{enumerate}	

As $x \in X$, it is obvious that~\ref{item:S1checki} holds. 	
As for~\ref{item:S1checkii}, observe that $X\in\choice[A]\cap 2^{X\cup Y}$. \lemref{lemma:setalphahat} implies $\choice[X\cup Y]=\choice[A]\cap 2^{X\cup Y}$
and thus $X\in\choice[X\cup Y]$.

Finally, consider an arbitrary $Z\in\choice[Y\cup\set x]$. Observe that $Y\in\choice[B]\cap 2^{Y\cup\set x}$. By another application of \lemref{lemma:setalphahat}, we get $\choice[Y\cup\set x]=\choice[B]\cap 2^{Y\cup\set x}$ and, therefore,
$Z\in\choice[B]$. Moreover, by choice of~$Y$, there is some $z\in X\cap (Y\cup\set x)$ such that $z\notin Z$. Hence, $X\cap(Y\cup\set x)\nsubseteq Z$, which proves~\ref{item:S1checkiii}. 
\end{proof}

By an \emph{(S2)-violation for~$R$} we understand a pair $(A,B)\in 2^U\times 2^U$ such that $B\subseteq A$ and for some $Y\in\choice[B]$, it is the case that $Z\cap B\nsubseteq Y$ for all $Z\in\choice[A]$.

\begin{lemma}\label{lemma:S2check}
		Let $R$ be a preference relation. If there exists an (S2)-violation for~$R$, then there exist $X,Y\in L$ and $x\in X$ such that $(X\cup Y,Y\cup\set x)$ is also an (S2)-violation for~$R$.
\end{lemma}

\begin{proof}
	Assume that $(A,B)$ is an (S2)-violation for $R$. 
	Then there is some $Y\in\choice[B]$ such that $Z\cap B\nsubseteq Y$ for all $Z\in\choice[A]$. As $\choice[A]\neq\emptyset$, there is some $X\in\choice[A]$ such that $X\cup Y$ is minimal with respect to set-inclusion, \ie $Z\cup Y\subsetneq X\cup Y$ for no $Z\in\choice[A]$. Obviously, $X,Y \in L$. By our assumption, $X\cap B\nsubseteq Y$ and we may assume the existence of some~$x\in X\setminus Y$.  We prove that $(X\cup Y,Y\cup\set x)$ is also an (S2)-violation for $R$, \ie
	\begin{enumerate}
		\item\label{item:S2checki} $Y\cup\set x\subseteq X\cup Y$,
		\item\label{item:S2checkii} $Y\in\choice[Y\cup\set x]$, and
		\item\label{item:S2checkiii} $Z\cap (Y\cup\set x)\nsubseteq Y$ for all $Z\in\choice[X\cup Y]$.
	\end{enumerate}
	
As $x \in X$, \ref{item:S2checki}  obviously holds. As for~\ref{item:S2checkii}, observe that~$Y\in\choice[B]\cap 2^{Y\cup\set x}$. \lemref{lemma:setalphahat} implies that  $\choice[Y\cup\set x]=\choice[B]\cap 2^{Y\cup\set x}$ and thus
$Y\in\choice[Y\cup\set x]$. 

Finally, consider an arbitrary $Z\in\choice[X\cup Y]$. Observe that 
$X\in\choice[A]\cap 2^{X\cup Y}$. Another application of \lemref{lemma:setalphahat} yields $\choice[X\cup Y]=\choice[A]\cap 2^{X\cup Y}$ and, therefore,
$Z\in\choice[A]\cap 2^{X\cup Y}$. Moreover, by choice of~$X$, we have $Z\setminus Y=X\setminus Y$ for all $Z\in\choice[A]\cap 2^{X\cup Y}$. In particular, $x\in Z$. Since $x\notin Y$, we obtain $Z\cap(Y\cup\set x)\nsubseteq Y$, which proves~\ref{item:S2checkiii}.
\end{proof}

We can utilize Lemmas~\ref{lemma:S1check} and \ref{lemma:S2check} to obtain a polynomial-time algorithm to check the substitutability of a preference relation.

\begin{theorem}
	It can be checked in time $O(\listsize^2\setsize^2 (\listsize + \s^2))$ whether a given preference relation is substitutable.
\end{theorem}
\begin{proof}
	To test substitutability, we need to check whether both (S1) and (S2) hold. This is equivalent to verifying that neither an (S1)-violation nor an (S2)-violation exists. 
	
	The algorithm works as follows. Instead of checking all possible violations of (S1) and (S2), we can use Lemmas~\ref{lemma:S1check} and \ref{lemma:S2check} to restrict our attention to only certain types of possible violations. The algorithm exhaustively checks all these types of possible violations.

Let us first consider the case of (S1).	To check (S1), we know from Lemma~\ref{lemma:S1check}, that we can restrict our attention to violations of the form $(X\cup Y,Y\cup\set x)$ for some $X,Y\in L$ and $x\in X$. Therefore, the maximum number of pairs we need to check is upper-bounded by ${\listsize\choose 2}u$. 
	
Verifying an (S1)-violation of type $(A,B)=(X\cup Y,Y\cup\set x)$ requires us to do the following:
	
	\begin{itemize}
	\item compute $C(A)$ which takes time $O(\listsize\setsize)$, 
	\item compute $C(B)$ which takes time $O(\listsize\setsize)$, and
	\item test the main condition: for all $X\in\choice[A]$ there is some $Y\in\choice[B]$ such that $X\cap B\subseteq Y$. Testing the condition takes time $O(\s^2 \setsize)$.
\end{itemize}
	
	Therefore, verifying a violation of type $(A,B)=(X\cup Y,Y\cup\set x)$ takes time 
	$$O(\listsize\setsize)+O(\listsize\setsize)+O(\s^2 \setsize)= O(\listsize\setsize + \s^2{\setsize}).$$
	
	The time needed to check whether an (S1)-violation exists is then equal to the maximum number of pairs we need to check multiplied by the time required to verify one (S1)-violation which equals
	$$O({\listsize\choose 2}  \setsize)  \times O(\listsize\setsize + \s^2{\setsize})=O(\listsize^2\setsize(\listsize\setsize + \s^2{\setsize})).$$


The same analysis holds for checking whether an (S2)-violation exists. Therefore there exists an algorithm which runs in time $O(2\listsize^2\setsize (\listsize\setsize + \s^2{\setsize}))=O(\listsize^2\setsize (\listsize\setsize + \s^2{\setsize}))=O(\listsize^2\setsize^2 (\listsize + \s^2))$ and tests the substitutability of a preference relation. 
\end{proof}

By letting $s=1$, we get the following.

\begin{corollary}
		It can be checked in time $O(\listsize^3\setsize^2)$ whether a given strict preference relation is substitutable.
\end{corollary}

%



\section{Conclusion}

We examined substitutability of preferences which may include indifferences. It was shown that this general notion of substitutability can be tested in time polynomial in the length of the preference relation in list form.  On the domain of strict preferences, the (worst case) asymptotic running time of the algorithm turns out to be slightly faster than the algorithm of \citet{HIK11a} ($O(\listsize^3 \setsize^2)$ as compared to $O(\listsize^3 \setsize^3)$). 
As pointed out by \citet{HIK11a}, \emph{``such an
algorithm could be distributed to market participants for use in the preparation
of their preference relations for submission.''} 

In contrast to other results in matching theory, allowing indifferences does therefore not affect the tractability of testing substitutability. 
It will be interesting to explore the extent to which allowing indifferences affects other recent results concerning substitutability. 

\section*{Acknowledgements}

This material is based on work supported by the Deutsche Forschungsgemeinschaft under grants BR 2312/6-1 (within the European Science Foundation's EUROCORES program LogICCC), BR~2312/7-1, and BR~2312/9-1.


\end{document}